\documentclass{article}

\usepackage[margin=1in]{geometry}
\usepackage{amsmath}
\usepackage{amsthm}
\usepackage{amssymb}
\usepackage{color}
\usepackage{mathrsfs}
\usepackage[algo2e]{algorithm2e}
\usepackage{tikz}
\usepackage{graphicx}
\usepackage{booktabs}
\usepackage{subcaption}
\usepackage{hyperref}

\newtheorem{theorem}{Theorem}
\newtheorem{corollary}[theorem]{Corollary}

\theoremstyle{remark}
\newtheorem{proposition}{Proposition}
\newtheorem{example}{Example}

\DeclareMathOperator*{\argmax}{arg\,max}
\DeclareMathOperator*{\argmin}{arg\,min}

\SetAlCapSkip{1em}

\title{A Center in Your Neighborhood: Fairness in Facility Location}
\author{Christopher Jung\\University of Pennsylvania \and Sampath Kannan\\University of Pennsylvania \and Neil Lutz\\Iowa State University}
\date{}

\begin{document}

\maketitle

\begin{abstract}
	When selecting locations for a set of facilities, standard clustering algorithms may place unfair burden on some individuals and neighborhoods. We formulate a fairness concept that takes local population densities into account. In particular, given $k$ facilities to locate and a population of size $n$, we define the ``neighborhood radius'' of an individual $i$ as the minimum radius of a ball centered at $i$ that contains at least $n/k$ individuals. Our objective is to ensure that each individual has a facility within at most a small constant factor of her neighborhood radius.

	We present several theoretical results: We show that optimizing this factor is NP-hard; we give an approximation algorithm that guarantees a factor of at most 2 in all metric spaces; and we prove matching lower bounds in some metric spaces. We apply a variant of this algorithm to real-world address data, showing that it is quite different from standard clustering algorithms and outperforms them on our objective function and balances the load between facilities more evenly.

\end{abstract}

\section{Introduction}
Fairness in decision making has become an important research topic as more and more classification decisions, such as college admissions, bank loans, parole and sentencing, 
are made with the assistance of machine learning algorithms~\cite{ChouRoth18}. Such decisions are made on individuals at a particular point in time, and although they have long-term consequences on the individuals affected, there is at least the prospect of these decisions being revisited with new data about these individuals. In contrast, certain infrastructural decisions, for example, about where to locate hospitals, schools, library branches, police stations, or fire stations have long-term consequences on all the residents of a town, district, or county.  Stories about neighborhoods not being adequately served make frequent headlines. Such stories range from food deserts in inner cities because of the absence of supermarkets that sell fresh fruit and vegetables to lack of access to medical services in rural areas~\cite{HaddonGasparro16,Virella18}.

What might be a fair way to locate $k$ hospitals, say, in a county with varying population densities, ranging from dense urban areas to sparse rural areas? 
The first approach one might think of would be to ask that everyone be within distance $x$ of their nearest hospital for the smallest possible $x$. This is the standard $k$-center formulation, but it is problematic from a fairness perspective for at least two reasons:

First, a hospital located in an urban area in a solution to this problem would likely be overused, since many people are likely to find this to be the hospital nearest to them. Thus, one kind of fairness we want is 
\textit{load balance}; the numbers of people served by each center should be as close to equal as possible. Intuitively the definition we give seems tailored to provide such balance, and we confirm this empirically.

Second, and perhaps more importantly, people living in areas with different population densities have different expectations for a reasonable distance to travel to a hospital. In rural areas, it would be unreasonable for a resident to expect to find a hospital within a mile, say, of her residence, but in an urban area this might be an entirely reasonable expectation. This is reinforced by the fact that individuals in dense urban areas --- especially dense, low-income urban neighborhoods --- are more likely to rely on bicycles or public transit and less likely to have access to a car~\cite{Trulia}. 

A good definition of fairness for such $k$-center or facility location problems should take into account population densities and geography. Consider the problem of serving a population $P$ of $n$ people
using  $k$ facilities, for a  given $k$. On average, we expect each facility to serve $n/k$ people. An individual $i$ might reasonably hope that the facility that serves $i$ is no farther than the $(\lceil \frac{n}{k} \rceil)$\textsuperscript{th} nearest individual from $i$, including $i$ itself. Thus, for a given $P$ and $k$, we define the \textit{neighborhood radius} $NR (i)$ to be the distance from $i$ to its $(\lceil \frac{n}{k} \rceil-1)$\textsuperscript{th} nearest neighbor.

It may not always be possible to find a solution with $k$ centers where each individual finds a center within her neighborhood radius, but our goal is to optimize how far we deviate from this ideal. Given a solution $S$ that specifies the placement of the $k$ centers, let $d(i,S)$ denote the distance from individual $i$ to the closest center in $S$. Let
\[\alpha (S) = \max _ i \frac{d(i,S)}{NR(i)}\]
denote the maximum factor by which an individual's distance to the center nearest to her, exceeds her neighborhood radius. We say that an algorithm achieves $\alpha$-fairness if the solution $S$ it produces has $\alpha (S) \le \alpha$.

The central problem we consider in this paper is:

\begin{description}
	\item[Fair $k$-Center:] Given $n$ points in a metric space, and a number $k$, find a solution $S^*$ consisting of a subset of at most $k$ of the given points so that \[S^* = \argmin _ {|S|\leq k} \alpha (S)\,.\]
\end{description}

One could formulate a Steiner version of this problem, where centers are allowed at arbitrary points in the  metric space, but we do not consider this variant in this paper.
We also formulate an extremal version of the problem: For a given metric space, what is the worst-case value, over all possible configurations of points in the metric space, of $\alpha(S^*)$?

We perform empirical comparisons between our fair $k$-center formulation and the standard $k$-center, $k$-means, and $k$-medians formulations. Using algorithms designed for each of these optimization problems, we select sets of facility locations based on two geographical data sets from Fairfax County, Virginia and Allegheny County, Pennsylvania.

Our results are as follows:

\begin{itemize}
\item There is an efficient algorithm that achieves $\alpha = 2$ for any set of points and any parameter $k$, in any metric space (Theorem~\ref{thm:general}).
\item Finding the optimal $\alpha$ for a given set of points and parameter $k$ is NP-hard (Theorem~\ref{thm:npcomplete}).
\item We prove that there are metric spaces and configurations of points for which $\alpha = 2$ is the best possible (Proposition~\ref{prop:generallower}). For Euclidean spaces we give a lower bound of $\alpha = \sqrt{2}$ (Proposition~\ref{prop:euclideanlower}).
\item On real data we show that standard clustering algorithms achieve worse $\alpha$ than is achieved by an algorithm we describe (Tables~\ref{table:fairfaxobjective} and~\ref{table:alleghenyobjective}).
\item We can associate with any algorithm a vector of at most $k$ values giving the number of points assigned to each center. Viewing load balance as the variance of this vector, 
we show empirically that our algorithm achieves much better load balance than the other clustering algorithms (Tables~\ref{table:clusterdev}).
\end{itemize}

\section{Defining $\alpha$-fairness}

We consider a nonempty collection $P$ of (not necessarily distinct) points in a metric space $(X,d)$ and some positive integer parameter $k\leq |P|$. A \emph{centers algorithm} takes an instance $(P,k)$ as input and returns a set $S\subseteq P$ with $|S|\leq k$ of designated \emph{centers}. The travel distance from a point $x\in X$ to $S$ is $d(x,S)$, the minimum distance from $x$ to a center in $S$:
\[d(x,S)=\min\{d(x,s):s\in S\}\,.\]

The goal of a centers algorithm is to select a ``good'' set of centers according to some criterion. For example, the following are well-studied optimization problems based on natural objective functions for assessing the quality of a solution set.

\begin{itemize}
\item[] \textbf{$k$-Center:} minimize the maximum travel distance among individuals in $P$, $\max_{i\in P}d(i,S)$~\cite{gonzalez1985clustering}.
\item[] \textbf{$k$-Medians:} minimize the average travel distance, or equivalently, $\sum_{i\in P}d(i,S)$~\cite{dubes1988algorithms}.
\item[] \textbf{$k$-Means:} minimize the sum of the squares of the travel distances, $\sum_{i\in P}d(i,S)^2$~\cite{macqueen1967some}.
\end{itemize}

The objective function we introduce, unlike those above, is based on the \emph{neighborhood radius} at a point $x\in X$, $NR(x)$. That is, the minimum radius $r$ such that at least $|P|/k$ of the points in $P$ are within distance $r$ of $x$:
\[NR_{P,k}(x)=\min\left\{r:\left\lvert B_r(x)\cap P\right\rvert\geq |P|/k\right\}\,,\]
where $B_r(x)$ is the closed ball of radius $r$ around $x$. When $P$ and $k$ are clear from context, we omit these subscripts and simply denote the neighborhood radius at $x$ by $NR(x)$.

We quantify the fairness of a set of centers $S$ on a set of points according to the worst ratio between travel distance and neighborhood radius for any point in $P$:
\[\alpha_{P,k}(S)=\sup_{i\in P}\frac{d(i,S)}{NR_{P,k}(i)}\,,\]
adopting the conventions that $0/0=\infty/\infty=1$ and $c/0=\infty$ for any $c>0$.

Given a centers algorithm $A$ and a constant $\alpha$, we say that $A$ achieves \emph{$\alpha$-fairness} on an instance $(P,k)$ if
\[\alpha_{P,k}(A(P,k))\leq \alpha\,.\] We say that $A$ is \emph{$\alpha$-fair} in the given metric space $(X,d)$ if it achieves $\alpha$-fairness on every instance. That is, if $\alpha_{P,k}(A(P,k))\leq \alpha$ for all $P\subseteq X$ and all $1\leq k\leq|P|$.

Solutions that are optimal for other standard objective functions can be infinitely unfair with respect to $\alpha$, as shown by the following example on the real line.
\begin{example}
Let $k=3$ and consider $P=\{-x,0,0,1,1,x\}$, where $x$ is some large number, as pictured in Figure~\ref{fig:different}. The optimal solution with respect to the $k$-center, $k$-medians, and $k$-means objective functions is to place one center at either 0 or 1 and the other two centers at $-x$ and $x$. But the neighborhood radius is 0 at 0 and 1, and whichever of these is not chosen as a center will have to travel a distance of 1 to the nearest center, meaning that
\[\alpha(\{-x,0,x\})=\alpha(\{-x,1,x\})\geq \frac10=\infty\,.\]

\begin{figure}[h]
	\begin{centering}
	\begin{tikzpicture}[dot/.style = {draw,fill,circle,thick,inner sep=0pt,minimum size=4pt}]
		\draw[thick,<->] (-4,0) -- (4,0);
		\node[dot] at (0,0.1) {};
		\node[dot] at (0,-0.1) {};
		\node[dot] at (0.5,0.1) {};
		\node[dot] at (0.5,-0.1) {};
		\node[dot] at (3.5,0) {};
		\node[dot] at (-3.5,0) {};
		\node at (0,-0.5) {0};
		\node at (0.5,-0.5) {1};
		\node at (3.5,-0.5) {$x$};
		\node at (-3.5,-0.5) {$-x$};
	\end{tikzpicture}
	\caption{For the above population with $k=3$, optimizing for $\max_{i\in P}d(i,S)$, $\sum_{i\in P}d(i,S)$, or $\sum_{i\in P}d(i,S)$ will yield a solution that is not $\alpha$-fair for any finite $\alpha$.}\label{fig:different}

	\end{centering}
\end{figure}
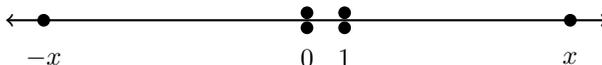
\end{example}

Although we do not consider allowing Steiner points as centers in this paper, it is clear that even optimal Steiner solutions to the three classical problems --- all of which place only one center in the interval $[0,1]$ --- do no better in terms of our fairness objective. This example demonstrates that a different approach is needed to achieve even the weakest of fairness guarantees. As a warm-up to exploring these approaches, we now give two examples of metric spaces in which strong fairness guarantees are easy to achieve.

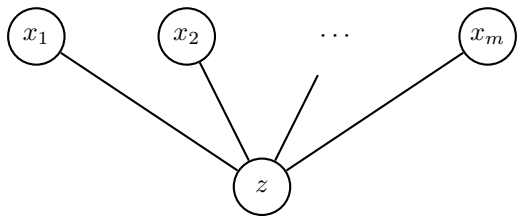
\begin{figure}[h]
	\begin{centering}
	\begin{tikzpicture}
	[every path/.style = {draw,thick},every node/.style = {inner sep=2pt, minimum size=0.75cm}]
		\node[draw,circle,thick] (z) at (0,0) {$z$};
		\node[draw,circle,thick] (x1) at (-3,2) {$x_1$};
		\node[draw,circle,thick] (x2) at (-1,2) {$x_2$};
		\node[minimum size=1cm] (dots) at (1,2) {$\cdots$};
		\node[draw,circle,thick] (xm) at (3,2) {$x_m$};
		\draw[thick] (z) -- (x1);
		\draw[thick] (z) -- (x2);
		\draw[thick] (z) -- (dots);
		\draw[thick] (z) -- (xm);
	\end{tikzpicture}
	\caption{1/2-fairness can be achieved in this instance by selecting $z$ as a center.}\label{fig:star}

	\end{centering}
\end{figure}

\begin{example}[A star graph]\label{ex:star} Consider the graph metric on a star graph: a graph with vertex set $X=\{z,x_1,\ldots, x_{m}\}$ and edge set $\{\{z,x_1\},\ldots,\{z,x_m\}\}$, as in Figure~\ref{fig:star}. If $(m+1)/2<k<m+1$, then $NR_{X,k}(z)=1$ and \[NR_{X,k}(x_1)=\cdots=NR_{X,k}(x_m)=2\,.\]
As long as $z\in S$, we have $d(z,S)=0$ and
\[d(x_1,S)=\cdots=d(x_m,S)=1\,,\]
so an algorithm that selects $z$ as a center achieves $1/2$-fairness for any such instance.
\end{example}

\begin{example}[The discrete metric]\label{ex:discrete} Consider any nonempty set $X$ under the discrete metric, where $d(x,y)=0$ if $x=y$ and 1 otherwise. Let $P\subseteq X$ be any nonempty set, let $1\leq k\leq |P|$, let $S\subseteq P$ be any nonempty set of size $k$, and let $i\in P$. If $k=|P|$, then we have $NR(i)=d(i,S)=0$. Otherwise, $NR(i)=1$ and $d(i,S)\in\{0,1\}$. Hence, there is a 1-fair algorithm for discrete metric spaces: Select any nonempty set of centers.
\end{example}

Most metric spaces do not share the essential property of Examples~\ref{ex:star} and~\ref{ex:discrete}: the existence of an extremely central point that is close to all other points. In general, achieving fairness requires more care about how one distributes multiple centers.

An ideal situation for the goal of 1-fairness would if the population were arranged in well-separated ``villages'' containing $n/k$ individuals each, where the diameter of each village is less than the space between villages. In this case, placing a single center anywhere in each village would achieve 1-fairness. 

But consider what happens if two villages are brought closer together: Some of an individual's $\lceil n/k\rceil -1$ closest neighbors might then reside in the other village, meaning that her neighborhood radius no longer encompasses the entirety of her own village, possibly including that village's center. This general situation is more difficult, but in the one-dimensional case, 1-fairness can still be achieved, as the following example shows.

\begin{example}[The real line] Given any finite set $P\subseteq\mathbb{R}$, Algorithm~\ref{alg:realline} starts from the left and takes every $\lceil n/k\rceil$\textsuperscript{th} point. Notice that the population $P$ in which the neighborhood radius is determined changes with each iteration. 
\begin{algorithm2e}
\SetAlgoLined
$S=\emptyset$\\
\While{$P\neq\emptyset$}{
	$s= \min P$\\
	$S=S\cup\{s\}$\\
	$P=P\setminus B_{NR_{P,k}}(s)$
}
\Return $S$\\
\caption{$\textsc{RealLineFairKCenter}(P,k)$}\label{alg:realline}
\end{algorithm2e}

Each iteration removes at least $n/k$ points from $P$, so the algorithm will terminate with at most $k$ centers. The $\lceil n/k\rceil$ closest points to any point on the line must include the $j\lceil n/k\rceil$\textsuperscript{th} smallest point for some $1\leq j\leq k$, so this algorithm is 1-fair.
\end{example}

Unfortunately, this approach cannot be extended to higher dimensions. As we show in Section~\ref{sec:theoretical}, 1-fairness is not always achievable, even in the Euclidean plane.

\section{Theoretical Results}\label{sec:theoretical}
Given an instance $(P,k)$, let $\alpha^*_{P,k}$ be the minimum value such that $\alpha$-fairness can be achieved. In this section we prove that
\[1/2\leq \alpha^*_{P,k}\leq 2\]
always holds, that equality is possible at each end of that bounding interval, and that $\alpha^*_{P,k}$ is NP-hard to compute.
\subsection{A 2-Fair Algorithm}
We now give an algorithm, \textsc{2FairKCenter}, that achieves 2-fairness on every instance and in every metric space. In each iteration, \textsc{2FairKCenter} chooses a center $s$ with minimum neighborhood radius among the set $Z$ of candidate centers. Then, it removes from $Z$ all points $i$ that are sufficiently close to $s$.

We have recently become aware that achieving 2-fairness is equivalent to finding a $(k/n)$-density net, as defined by Chan, Dinitz, and Gupta in the context of constructing slack spanners~\cite{ChDiGu06}. In proving that $\epsilon$-density nets can be found in polynomial time for all $\epsilon \in (0,1)$, that work describes an algorithm that is essentially identical to \textsc{2FairKCenter}. In order to keep this paper self-contained, we include the algorithm description and proof of 2-fairness here. 

\begin{algorithm2e}
\SetAlgoLined
$Z=P$\\
$S=\emptyset$\\
\While{$S\neq\emptyset$}{
	choose $s\in\argmin_{i\in Z}NR_{P,k}(i)$\\
	$S=S\cup\{s\}$\\
	$Z=\{i\in Z:d(i,s)>NR_{P,k}(i)+NR_{P,k}(s)\}$\\}	
\Return $S$
\caption{$\textsc{2FairKCenter}(P,k)$}\label{alg:general}
\end{algorithm2e}

\begin{theorem}\label{thm:general}
	\textsc{2FairKCenter} is $2$-fair in every metric space.
\end{theorem}

\begin{proof}
Fix a metric space $(X,d)$, let $P\subseteq X$, and let $k\leq |P|$. Let $s_j$ be the $j$\textsuperscript{th} center added to $S$. For any $j'>j$, the definition of the set $S$ guarantees that $B_{NR(s_j)}(s_j)$ and $B_{NR(s_{j'})}(s_{j'})$ are disjoint. Thus, the balls
\[B_{NR(s_1)}(s_1),B_{NR(s_2)}(s_2),\ldots\]
are all pairwise disjoint, and by the definition of neighborhood radius, each includes at least $|P|/k$ points in $P$. It follows that there can be at most $k$ centers, so this algorithm will output a valid solution.

Now, a point $i\in P$ is excluded from $Z$ only when there is some $s\in S$ such that $d(i,s)\leq NR(i)+NR(s)$, which is at most $2\cdot NR(i)$ by our choice of $s$. So when our algorithm terminates, $d(i,S)/NR(i)\leq 2$ holds for all $i\in P$. Thus, 2-fairness is achieved on $(P,k)$, and the algorithm is 2-fair.
\end{proof}

\subsection{Lower Bounds}\label{ssec:lowerbounds}
We now give four lower bounds on fairness: We prove that it is never possible to achieve better than 1/2-fairness on any instance; that no algorithm can be better than 1-fair, regardless of the metric space; that there exist metric spaces in which no algorithm can be better than 2-fair; and that no algorithm can be better than $\sqrt{2}$-fair in Euclidean spaces of dimension greater than 1. The first three of these results demonstrate the tightness of Example~\ref{ex:star}, Example~\ref{ex:discrete}, and Theorem~\ref{thm:general}, respectively.

\begin{proposition}\label{prop:instancelower}
	In every metric space $(X,d)$, for all $S\subseteq P\subseteq X$ and $1\leq k\leq |S|$, we have $\alpha_{P,k}(S)\geq 1/2$.
\end{proposition}
\begin{proof}
	For each center $s\in S$, define the set
	\[J(s)=\{i\in P:d(i,s)=d(i,S)\}\,,\]
	the set of points in $P$ for which $s$ is a closest center. There are at most $k$ centers in $S$, and $\bigcup_{s\in S}J(s)=P$, so there must be some center $\hat{s}$ with $J(\hat{s})\geq |P|/k$. Now, let \[\hat{i}\in\argmax_{i\in J(\hat{s})}d(i,\hat{s})\,.\]
	For each $j\in J(\hat{s})$, we know that $d(j,\hat{s})\leq d(\hat{i},\hat{s})$, so by the triangle inequality, $d(\hat{i},j)\leq 2d(\hat{i},\hat{s})$. Thus, $B_{2d(\hat{i},\hat{s})}(\hat{i})$ contains all of $J(\hat{s})$, meaning that
	\[\big|B_{2d(\hat{i},\hat{s})}\cap P\big|\geq |J(\hat{s})|\geq |P|/k\,,\]
	so $NR_{P,k}(\hat{i})\leq 2d(\hat{i},\hat{s})=2d(\hat{i},S)$. It follows that $\alpha_{P,k}(S)\geq 1/2$.
\end{proof}

Combining Theorem~\ref{thm:general} with Proposition~\ref{prop:instancelower} immediately yields the following.
\begin{corollary}
	\textsc{2FairKCenter} is a 4-approximation algorithm.
\end{corollary}

\begin{proposition}
	For all metric spaces $(X,d)$ and all $\alpha<1$, there is no centers algorithm that is $\alpha$-fair in $(X,d)$.
\end{proposition}
\begin{proof}
Suppose the number of available centers is the same as the size of the population: $k=|P|$. Then $NR_{P,k}(x)=0$ for all $x$, so $\alpha_{P,k}(A(P,k))$ is $1$ if $A$ places a center at every point in $P$ and $\infty$ otherwise.
\end{proof}

\begin{proposition}\label{prop:generallower}
	There exists a metric space $(X,d)$ such that for all $\alpha<2$, there is no centers algorithm that is $\alpha$-fair in $(X,d)$.
\end{proposition}
\begin{figure}[h]
	\begin{centering}
	\begin{tikzpicture}
	[every path/.style = {draw,thick},every node/.style = {draw,fill,circle,thick,inner sep=0pt,minimum size=5pt}]
		\node (a) at (0,0) {};
		\node (b) at (1,0) {};
		\node (c) at (3,0) {};
		\node (d) at (4,0) {};
		\node (e) at (6,0) {};
		\node (f) at (7,0) {};
		\node (g) at (0,1) {};
		\node (h) at (1,1) {};
		\node (i) at (3,1) {};
		\node (j) at (4,1) {};
		\node (k) at (6,1) {};
		\node (l) at (7,1) {};
		\path (a)--(b)--(h)--(g)--(a);
		\path (c)--(d)--(j)--(i)--(c);
		\path (e)--(f)--(l)--(k)--(e);	
	\end{tikzpicture}
	\caption{Under a graph metric, it is impossible to do better than 2-fairness when choosing four centers from among this set of points.}\label{fig:squares}

	\end{centering}
\end{figure}
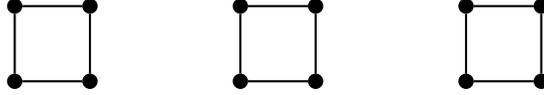
\begin{proof}
	Consider the example in Figure~\ref{fig:squares}, with 12 points under a graph metric with $k=4$ and $P=X$. For every point $i\in P$, we have $NR(i)=1$. But for any choice of three centers, some square will have at most one center, and one point in that square will therefore have travel distance at least 2. Thus, no centers algorithm in this metric space can be $\alpha$-fair for any $\alpha<2$.
\end{proof}

\begin{proposition}\label{prop:euclideanlower}
	For all $m\geq 2$ and all $\alpha<\sqrt{2}$, there is no centers algorithm that $\alpha$-fair in $m$-dimensional Euclidean space.
\end{proposition}
\begin{proof}
	This holds by essentially the same example used to prove Proposition~\ref{prop:generallower}: $P$ consists of 12 points arranged in three squares of unit side, where the distance between the squares is greater than the diameter of the squares, and $k=4$. Once again, every point has neighborhood radius 1, and for any solution $S$, some square will have at most one center. Some other point $i$ in that square will have travel distance $d(i,S)\geq\sqrt{2}$, and it follows immediately that $\alpha(S)\geq \sqrt{2}$.
\end{proof}

\subsection{NP-Completeness}
The $k$-center, $k$-medians, and $k$-means problems are all known to be NP-hard, and the problem of checking, for a given instance, whether a given value of the objective function is achievable, is NP-complete~\cite{DrineasFKVV04, MegiddoS84, Hochbaum84}. We now show that the same is true for our objective function $\alpha$.
\begin{theorem}\label{thm:npcomplete}
	The problem of determining whether $1$-fairness can be achieved on a given instance is NP-complete.
\end{theorem}
\begin{proof}
	This problem is a special case of the set cover problem, where the sets are
	\[S_i=\left\{j\in P:d(i,j)\leq \alpha\cdot NR_{P,k}(j)\right\}\]
	for each $i\in P$, so it belongs to NP.

	We prove NP-hardness by reduction from the dominating set problem. Let $G$ be a graph on a set $U$ of $n$ vertices, and let $1\leq k\leq n$; without loss of generality, we assume that $n-k$ is even. We construct a new graph $G'$ that contains $G$ as a subgraph and also has the following:
	\begin{itemize}
		\item a set $V$ of $2n$ vertices with degree 1 such that each vertex $u\in U$ is adjacent to two vertices $u_1,u_2\in V$, and
		\item a set $W$ of $6n-6k$ vertices arranged as $\frac{3}{2}(n-k)$ disjoint 4-cycles.
	\end{itemize}
	Letting $P=U\cup V\cup W$ and $k'=3n-2k$, we will show that $G$ has a dominating set of size $k$ if and only if there is a set $S\subseteq P$ with $|S|=k'$ and $\alpha_{P,k'}(S)\leq 1$. Since $(G',k')$ can be efficiently computed from $(G,k)$, this will suffice to prove the theorem.

	Suppose that $G$ has a dominating set $D$ of size $k$, and consider the set of centers $S=D\cup T$, where $T$ is a set consisting of two vertices from each of the squares in $W$, so that $d(w,S)=d(w,T)\leq 1$ for all $w\in W$. Now,
	\[\frac{|P|}{k'}=\frac{n+2n+6n-6k}{3n-2k}=3\,,\]
	so $NR_{P,k'}(w)=1$ for each $w\in W$. Each $u\in U$ has at least two neighbors --- namely, $u_1$ and $u_2$ --- so we also have $NR_{P,k'}(u)=1$ for all $u\in U$, and it follows immediately that $NR_{P,k'}(v)=2$ for each $v\in V$. The fact that $D\subseteq S$ is a dominating set means that $d(u,S)\leq 1$ for each $u\in U$ and therefore that $d(v,S)\leq 2$ for each $v\in V$. Thus, $\alpha_{P,k'}(S)=1$.

	Conversely, suppose that there is some set $S\subseteq P$ with $|S|=k'$ and $\alpha_{P,k'}(S)\leq 1$. Then $S$ must contain at least two vertices from each square in $W$, so letting $Y=S\cap(U\cup V)$, we have
	\[|Y|\leq k'-(3n-3k)=k\,.\]
	For each $u\in U$, we must have
	\[d(u,Y)=d(u,S)\leq NR_{P,k'}(u)=1\,.\]
	We construct a set $D$ by taking each vertex in $Y\cap V$ and replacing it with the adjacent vertex in $U$. Then $|D|\leq |Y|\leq k$, and for each $u\in U$, $d(u,D)\leq d(u,Y)$, meaning that $D$ is a dominating set for $G$. We conclude that this problem is NP-complete.
\end{proof}

\section{Experiments}
In this section we measure the fairness of three standard clustering algorithms on two geographic data sets, and we compare their performance to that of a modified version of \textsc{2FairKCenter} that still guarantees 2-fairness but attempts to be even fairer.
\subsection{A Heuristic Refinement of \textsc{2FairKCenter}}
Our algorithm \textsc{2FairKCenter} always yields a 2-fair solution, but this solution might be less than optimal and use fewer than $k$ centers. To avoid this situation, we introduce \textsc{AlphaFairKCenter}, a version of \textsc{2FairKCenter} that is parameterized by a fairness guarantee parameter, $\alpha$. This algorithm achieves $\alpha$-fairness on every instance by essentially the same argument we used to prove that Algorithm~\ref{alg:general} is 2-fair. The catch is that the output will not necessarily be a valid solution: For $\alpha<2$, Algorithm~\ref{alg:parameterized} may select more than $k$ centers on a given instance $(P,k)$.

\begin{algorithm2e}
\SetAlgoLined
$Z=P$\\
$S=\emptyset$\\
\While{$S\neq\emptyset$}{
	choose $s\in\argmin_{i\in Z}NR_{P,k}(i)$\\
	$S=S\cup\{s\}$\\
	$Z=\{i\in Z:d(i,s)>\alpha\cdot NR_{P,k}(i)\}$	
}
\Return $S$
\caption{$\textsc{AlphaFairKCenter}(\alpha,P,k)$}\label{alg:parameterized}
\end{algorithm2e}

\begin{figure}[h]
	\begin{subfigure}{0.48\textwidth}
		\includegraphics[width=\linewidth]{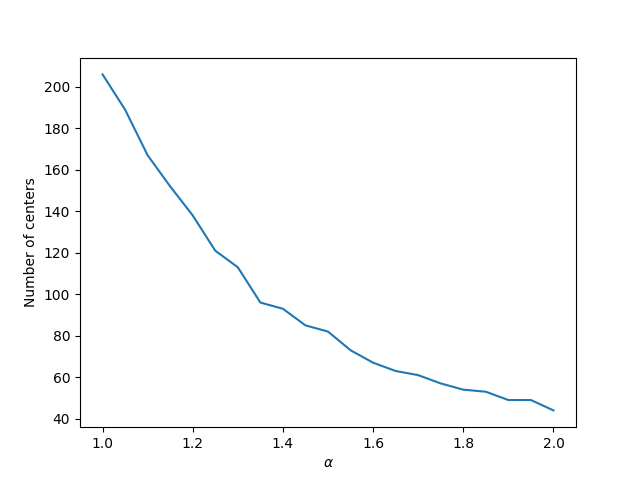}
		\caption{Fairfax County, Virginia}
	\end{subfigure}
	\hfill
	\begin{subfigure}{0.48\textwidth}
		\includegraphics[width=\linewidth]{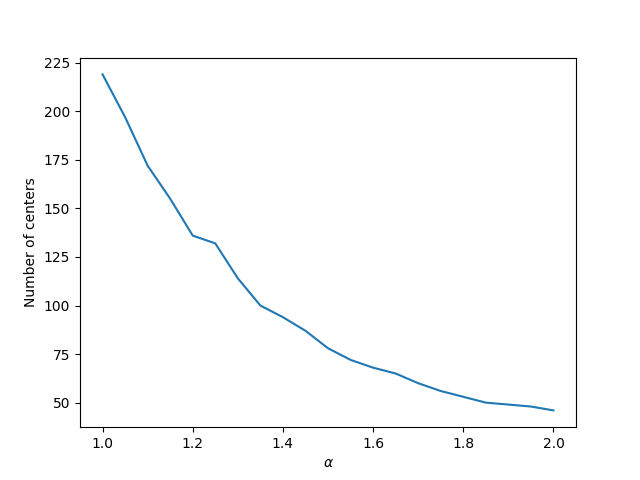}
		\caption{Allegheny County, Pennsylvania}
	\end{subfigure}
	\caption{The number of centers chosen by \textsc{AlphaFairKCenter} for different values of $\alpha$ with $k=100$ given address points in two counties.}\label{fig:monotonic}
\end{figure}

For each instance $(P,k)$, we define a function $f_{P,k}:[1,2]\to\mathbb{N}$ by \[f_{P,k}(\alpha)=\left|\textsc{AlphaFairKCenter}(\alpha,P,k)\right|\,,\]
the number of centers chosen by \textsc{AlphaFairKCenter} with parameter $\alpha$ on instance $(P,k)$. Our goal is to find a small $\alpha$ such that $f_{P,k}(\alpha)\leq k$. In order to do this, we will perform a binary search on the interval $[1,2]$, recursively searching the lower half of the interval when $f_{P,k}(\alpha)> k$ and the upper half of the interval otherwise.

If $f_{P,k}$ is a monotonic function, then this search will find
\[\inf \{\alpha\in[1,2]:f_{P,k}(\alpha)\leq k\}\]
up to arbitrary precision. Intuitively, $f_{P,k}$ has a general tendency to be decreasing --- a weaker fairness guarantee requires fewer centers --- but in fact $f_{P,k}$ is not necessarily monotonic, and local extrema may cause our search to select a larger $\alpha$ than is necessary.

Fortunately, as shown in Figure~\ref{fig:monotonic}, $f_{P,k}$ seems to behave monotonically at coarse scales on real data. Furthermore, deviations from monotonicity cannot affect the validity of the solution we find, only its optimality. Hence, this binary search appears to be a useful heuristic, and we employ it in our algorithm $\textsc{FairKCenter}$. In addition to an instance $(P,k)$, this algorithm takes as input a precision parameter $t$ that determines the depth of the binary search.

\begin{algorithm2e}[h]
$low = 1$\\
$high = 2$\\
\For{$i=1,2,\ldots,t$}{
	$mid = (low + high)/2$\\
	\If{$|\textsc{AlphaFairKCenter}(mid,P,k)|\leq k$}{
		$high = mid$\\
	}
	\Else{$low = mid$}
}
\Return $\textsc{AlphaFairKCenter}(high,P,k)$
\caption{$\textsc{FairKCenter}(t,P,k)$}\label{alg:binarysearch}
\end{algorithm2e}

\subsection{Experimental Setup}

We applied our algorithm \textsc{FairKCenter} to select 100 center locations in two American counties: Fairfax County, Virginia, and Allegheny County, Pennsylvania. Fairfax County is located near Washington, D.C., and is primarily suburban. According to the 2010 United States Census~\cite{Census}, its population density is 1068 people per square kilometer, with census tracts ranging in density from 56 to 23,397 people per square kilometer. Allegheny County contains the city of Pittsburgh as well as many of its suburbs and exurbs; in the 2010 Census, the county's population density was 647 people per square kilometer, with census tracts ranging in density from 48 to 12,474 people per square kilometer.\footnotemark

\footnotetext{The stated ranges of population density exclude the few census tracts with fewer than 100 people. Some census tracts are uninhabited.}

The ``populations'' for our experiment were the sets of all address points in each county, not the locations of individual people. The Fairfax data set contains 537,514 address points, and the Allegheny data set contains 370,776. The data sets were published by Fairfax County GIS and the Allegheny County / City of Pittsburgh / Western PA Regional Data Center, respectively~\cite{FairfaxData,AlleghenyData}. We measured Euclidean distance after projecting (latitude, longitude) pairs onto the plane using the Universal Transverse Mercator coordinate system.

The bottleneck for our algorithm in terms of running time is calculating the neighborhood radius for each point. In order to accelerate this process, we used the Python library \href{https://scikit-learn.org/stable/modules/generated/sklearn.neighbors.KDTree.html}{KD-tree}~\cite{bentley1975multidimensional}. The KD-tree data structure allows us to quickly query the distance to any point's $(\lceil n/k\rceil - 1)$\textsuperscript{th} nearest neighbor, which is exactly the definition of the neighborhood radius at that point.

\begin{figure*}[p]
\begin{subfigure}{0.48\textwidth}
\includegraphics[width=\linewidth]{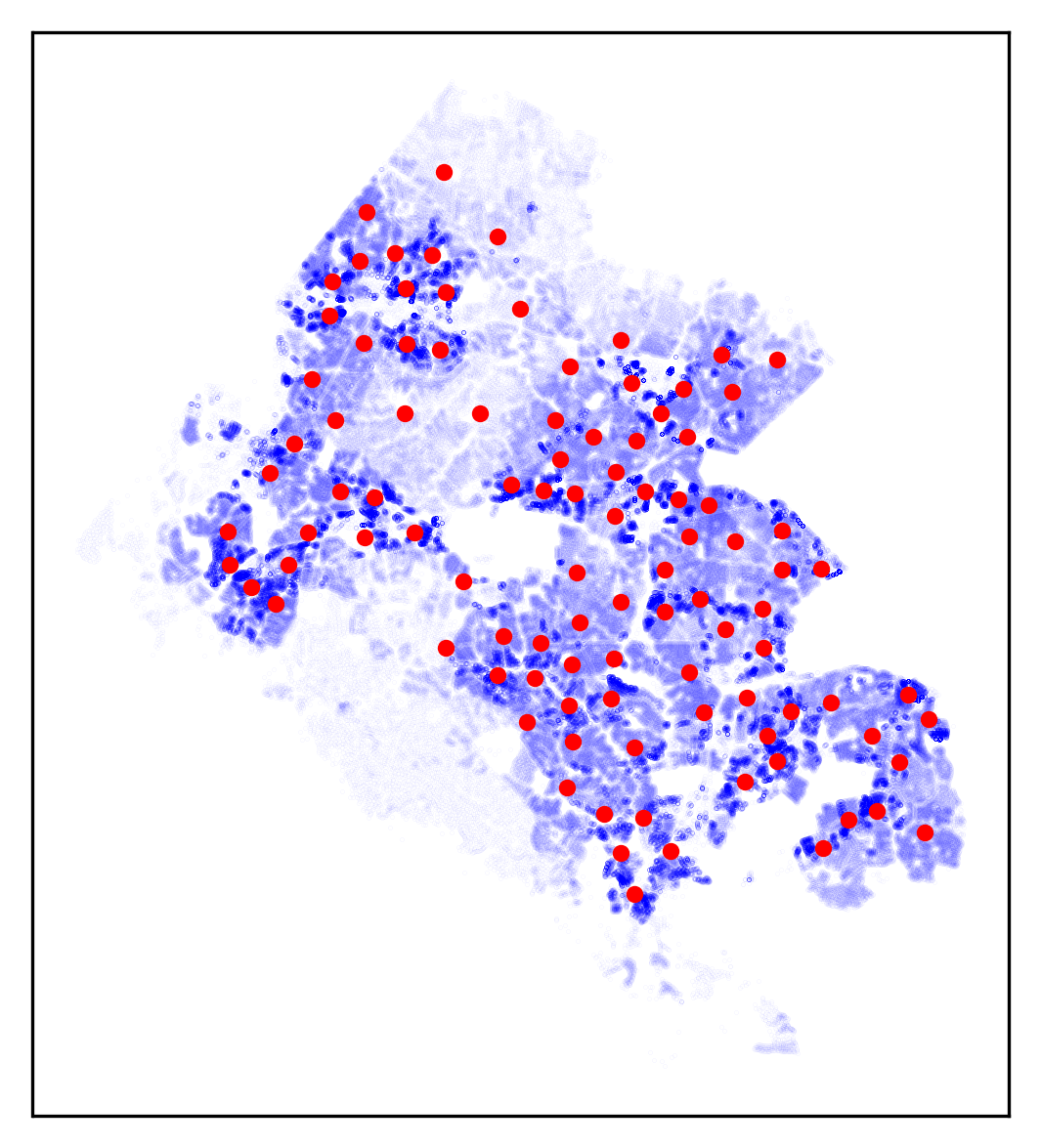}
\caption{\textsc{FairKCenter}}
\end{subfigure}
\hfill
\begin{subfigure}{0.48\textwidth}
\includegraphics[width=\linewidth]{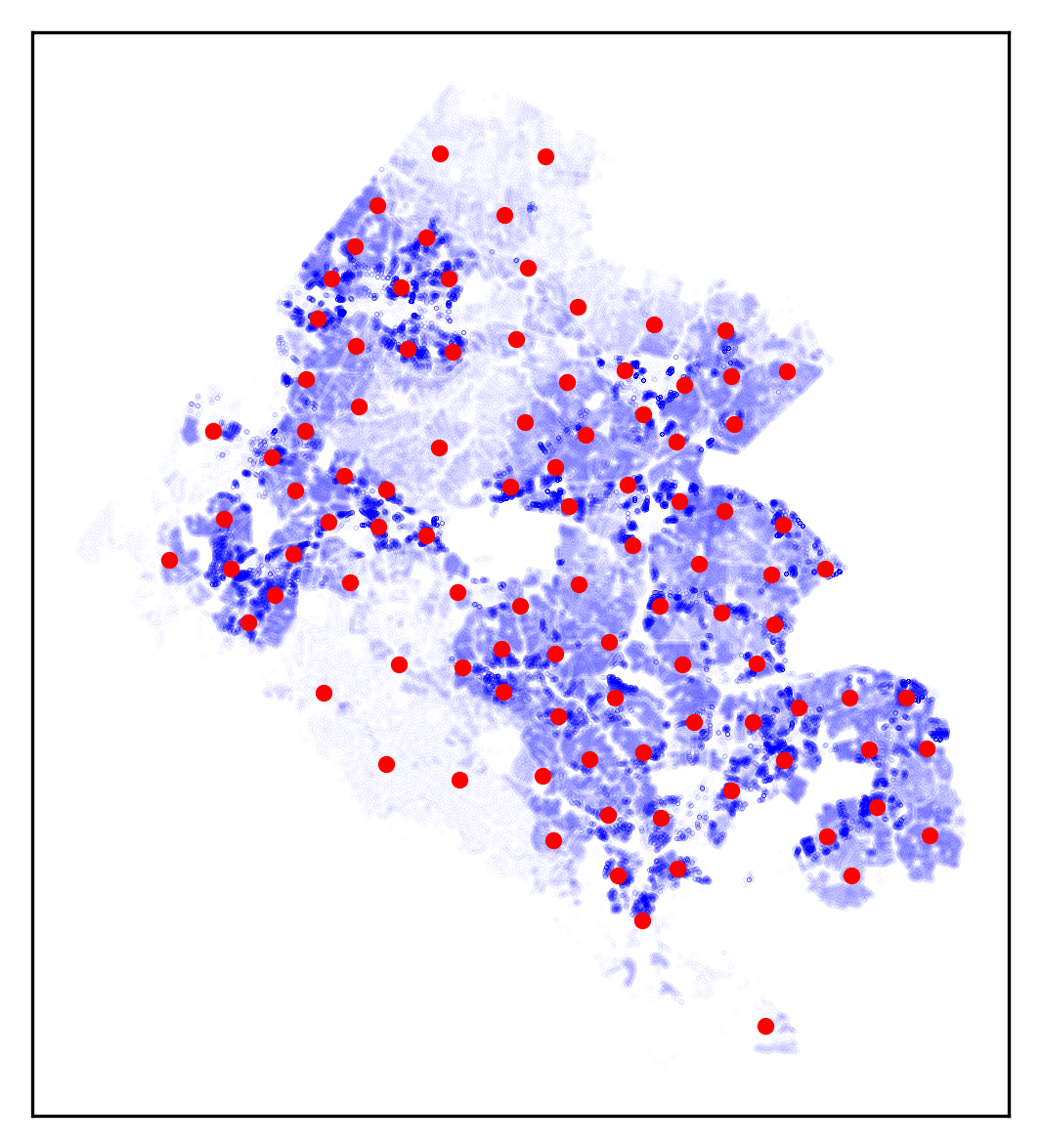}
\caption{$k$-means}
\end{subfigure}

\bigskip
\begin{subfigure}{0.48\textwidth}
\includegraphics[width=\linewidth]{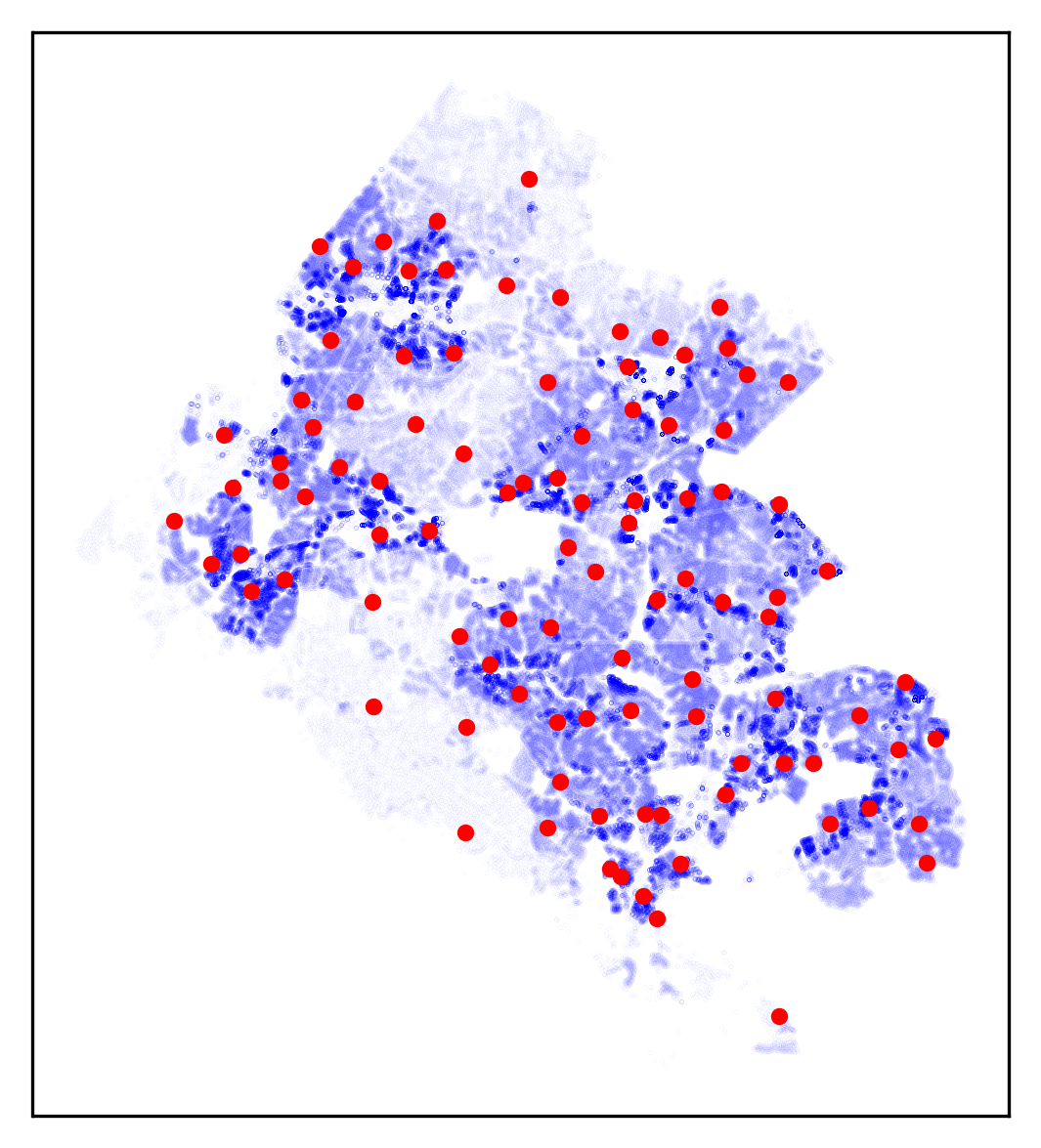}
\caption{$k$-medians}
\end{subfigure}
\hfill
\begin{subfigure}{0.48\textwidth}
\includegraphics[width=\linewidth]{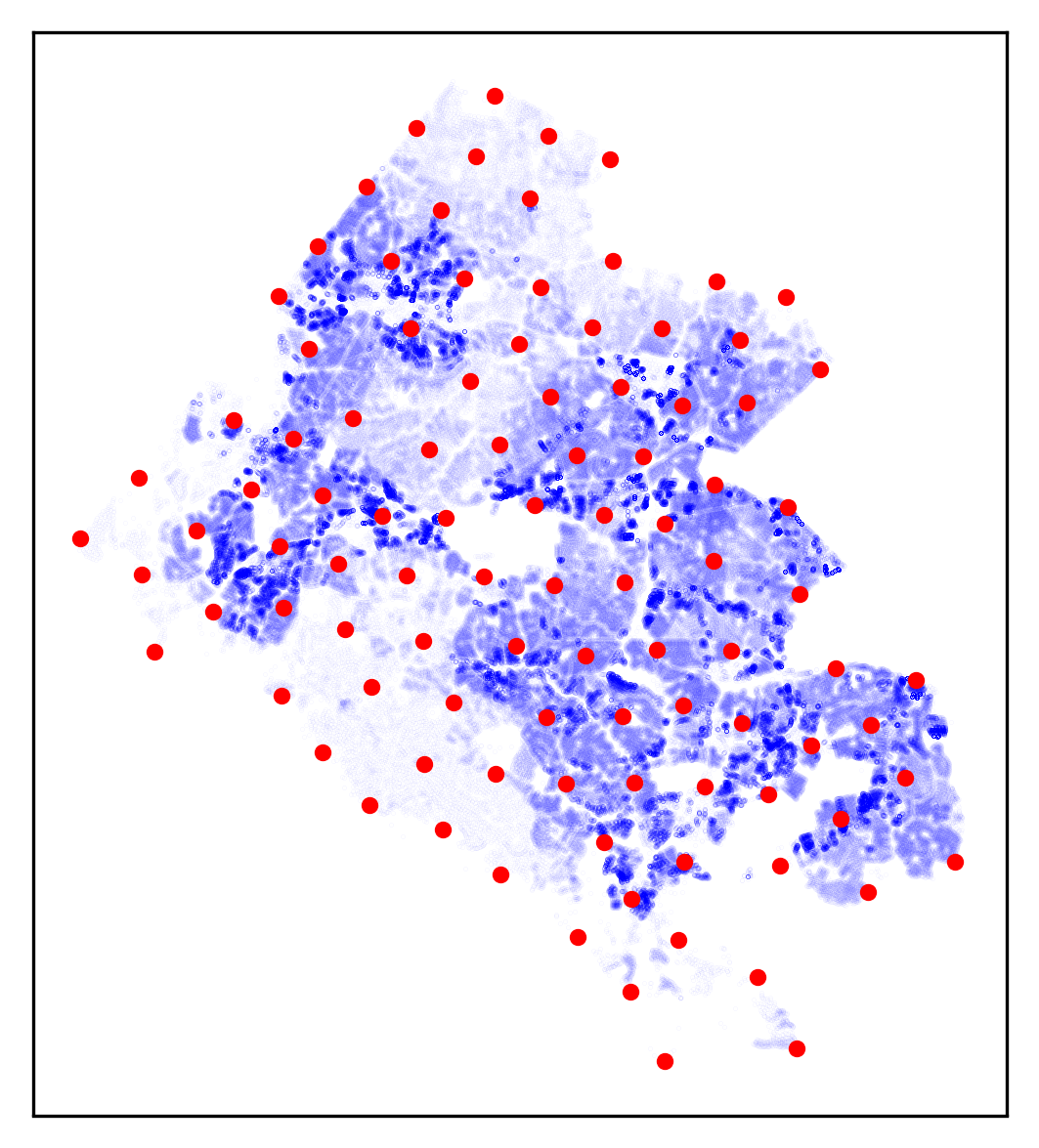}
\caption{$k$-center}
\end{subfigure}

\caption{Placing 100 centers in Fairfax County, using $\textsc{FairKCenter}$ and algorithms for the $k$-means, $k$-medians, and $k$-center problems.}\label{fig:fairfaxmaps}
\end{figure*}

\begin{figure*}[hp]
\begin{subfigure}{0.48\textwidth}
\includegraphics[width=\linewidth]{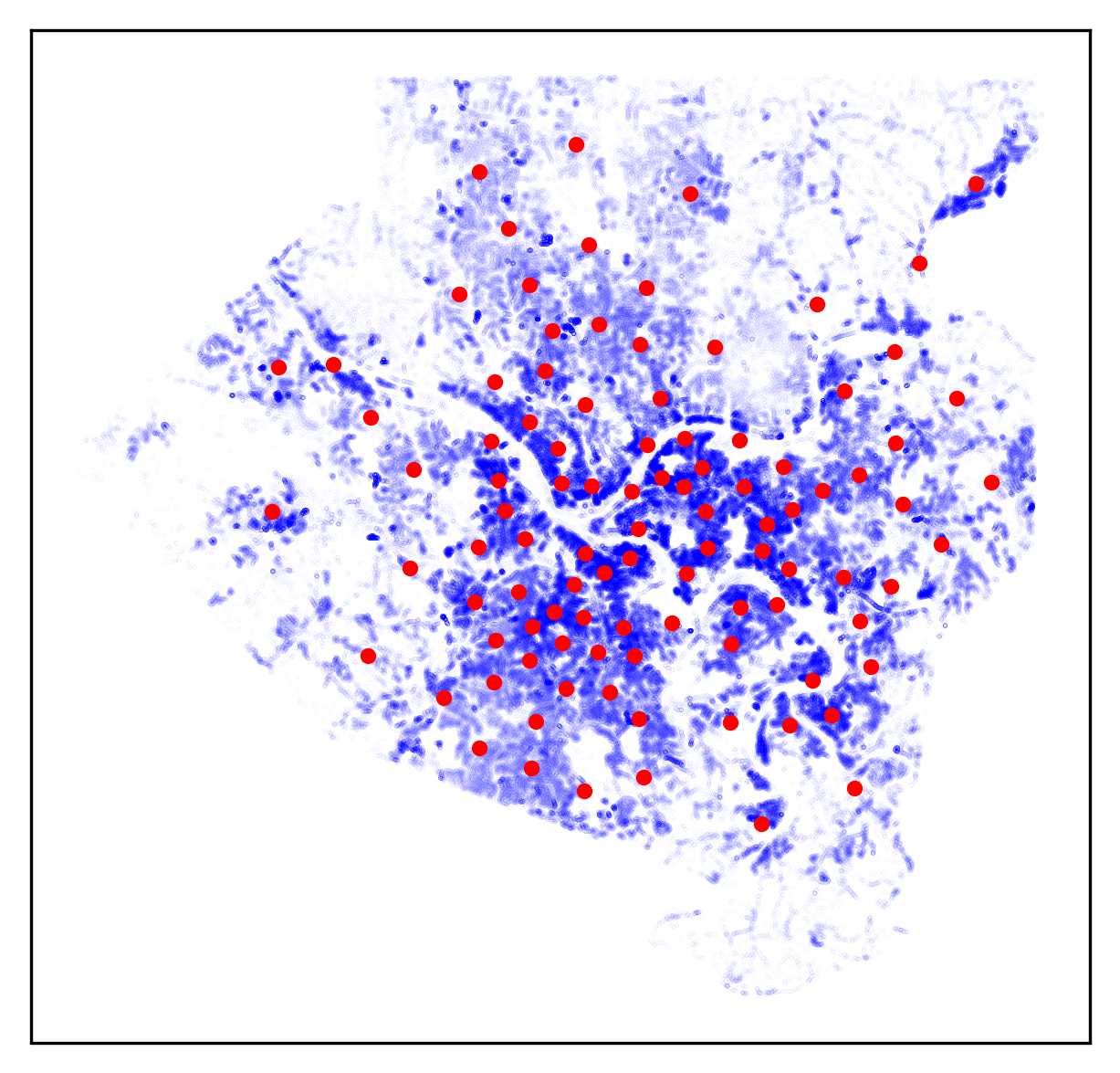}
\caption{\textsc{FairKCenter}}
\end{subfigure}
\hfill
\begin{subfigure}{0.48\textwidth}
\includegraphics[width=\linewidth]{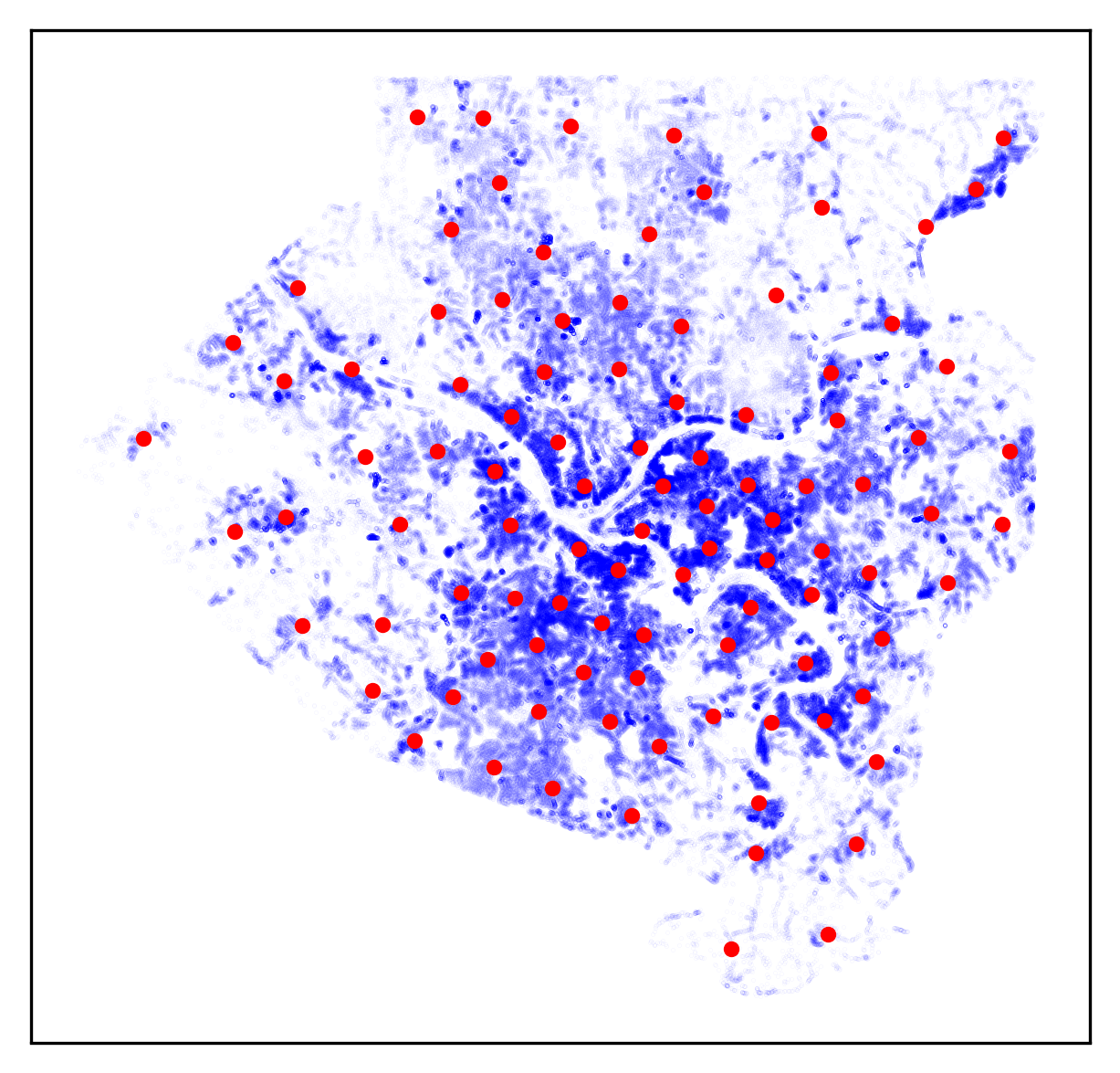}
\caption{$k$-means}
\end{subfigure}

\bigskip
\begin{subfigure}{0.48\textwidth}
\includegraphics[width=\linewidth]{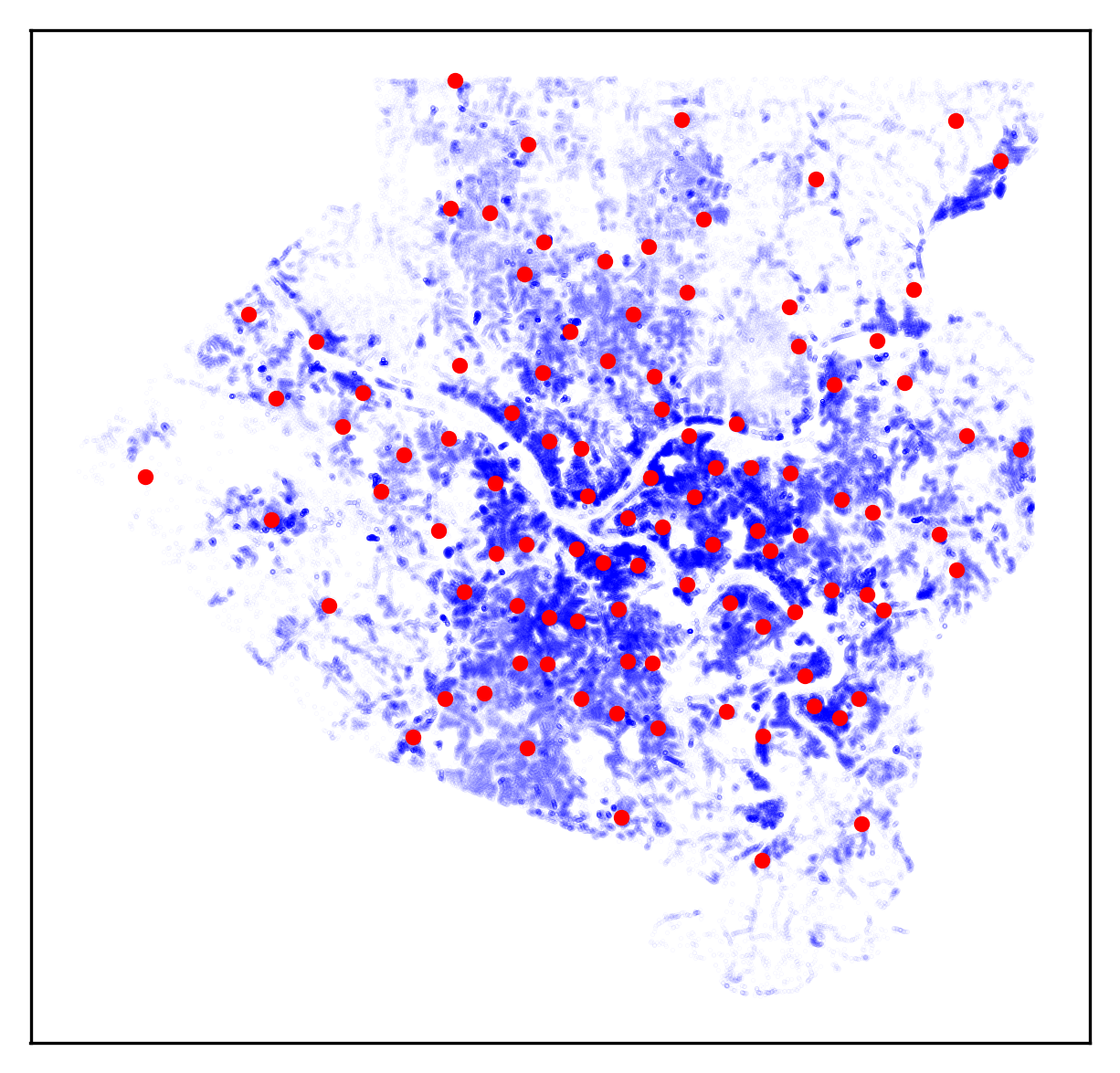}
\caption{$k$-medians}
\end{subfigure}
\hfill
\begin{subfigure}{0.48\textwidth}
\includegraphics[width=\linewidth]{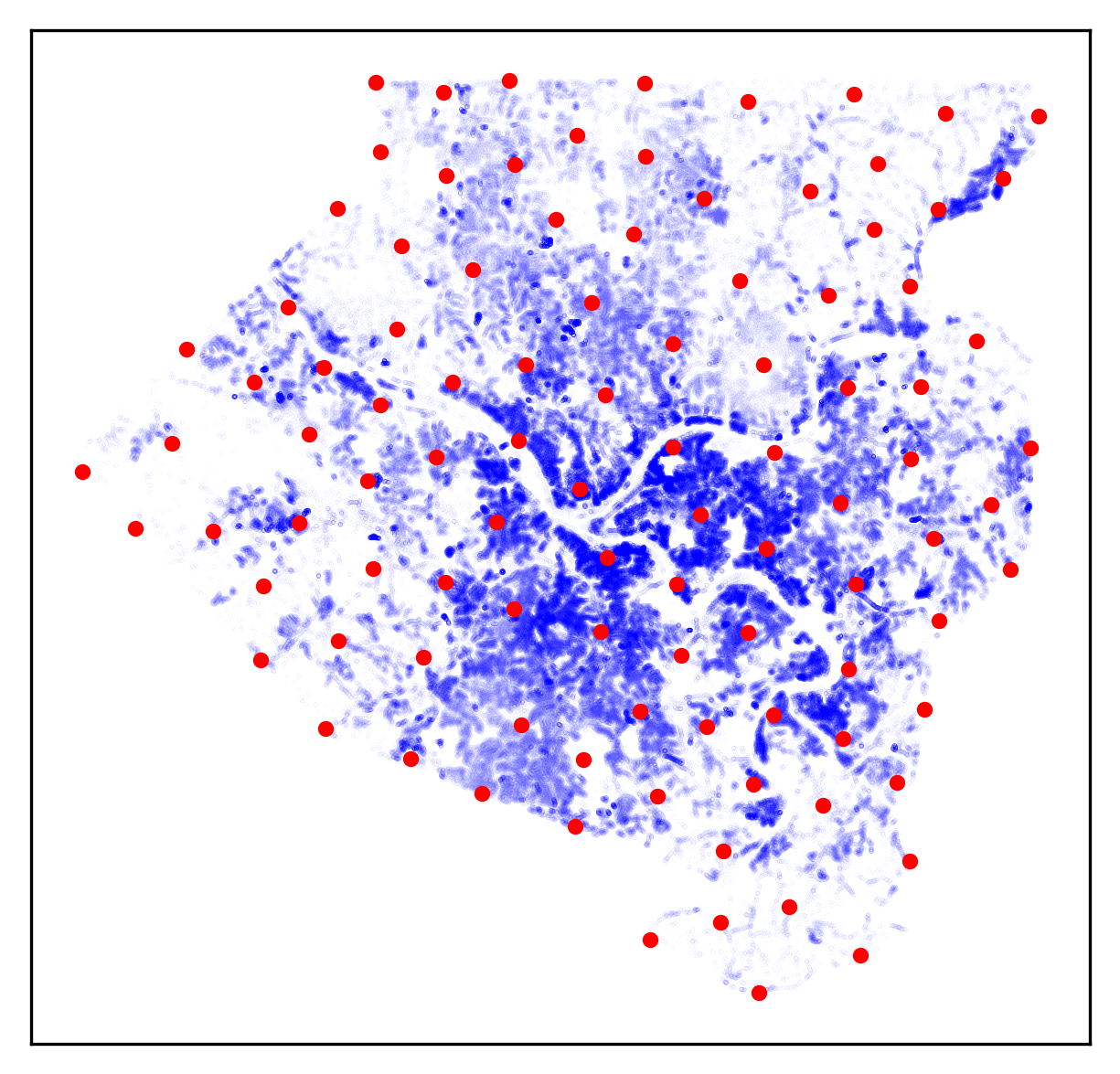}
\caption{$k$-center}
\end{subfigure}

\caption{Placing 100 centers in Allegheny County, using $\textsc{FairKCenter}$ and algorithms for the $k$-means, $k$-medians, and $k$-center problems.}\label{fig:alleghenymaps}
\end{figure*}

We compared the performance of \textsc{FairKCenter} to standard algorithms for the $k$-means, $k$-medians, and $k$-center problems:
\begin{itemize}
	\item For $k$-means, we used the Python library \href{https://scikit-learn.org/stable/modules/generated/sklearn.cluster.KMeans.html}{sklearn}, which employs either Lloyd's algorithm~\cite{lloyd1982least} or Elkan's algorithm~\cite{elkan2003using}, depending on the problem size and parameters.
	\item For $k$-medians, we used the Python library \href{https://codedocs.xyz/annoviko/pyclustering/classpyclustering_1_1cluster_1_1kmedians_1_1kmedians.html}{pyclustering} to execute a variant of Lloyd's algorithm that calculates a median instead of a centroid in each iteration.
	\item For $k$-center, we implemented the standard greedy approximation algorithm~\cite{williamson2011design}.
\end{itemize}

For each county, we assessed the performance of each algorithm according to our fairness objective function $\alpha$ as well as the $k$-means, $k$-medians, and $k$-center objective functions. We also measured how well each algorithm balanced the load by finding the standard deviation in the number of addresses served by each center.

\subsection{Experimental Results}

In Figure~\ref{fig:fairfaxmaps}, we show the population density map of Fairfax County along with 100 centers whose locations were determined by \textsc{FairKCenter}, $k$-means, $k$-medians, and $k$-center. Each tiny blue point, whose transparency has been slightly lowered in order to better show the population density of each region, corresponds to an address point, and each red point is a center. We do the same for Allegheny County in Figure~\ref{fig:alleghenymaps}.

\begin{table}[htp]
\begin{centering}
\caption{Performance of each algorithm on Fairfax County with respect to various objective functions}\label{table:fairfaxobjective}
\begin{tabular}{lllll}
\toprule
&\multicolumn{4}{c}{Objective function}\\ \cmidrule{2-5}
Algorithm              & $\alpha$ & $k$-means & $k$-medians & $k$-center \\
\midrule
\textsc{FairKCenter}			& 1.34306       & 1811007 & 1373.79   & 10002.64  \\
$k$-means		    			& 1.45643       & 1137163 & 1217.07   & 5662.06  \\ 
$k$-medians    					& 1.80263       & 1613910 & 1393.39   & 6446.81  \\ 
$k$-center  			  		& 2.57986       & 2027176 & 1675.85   & 2925.80  \\
\bottomrule
\end{tabular}

\end{centering}
\end{table}

\begin{table}[htp]
\begin{centering}
\caption{Performance of each algorithm on Allegheny County with respect to various objective functions}\label{table:alleghenyobjective}
\begin{tabular}{lllll}
\toprule
&\multicolumn{4}{c}{Objective function}\\ \cmidrule{2-5}
Algorithm              & $\alpha$ & $k$-means & $k$-medians & $k$-center \\
\midrule
\textsc{FairKCenter}			& 1.33721       & 3600461 & 1902.78   &11615.59  \\
$k$-means		    			& 1.57453       & 2082104 & 1632.00   & 6183.67  \\ 
$k$-medians    					& 1.90726       & 3020040 & 1841.34   & 7835.79  \\ 
$k$-center   			  		& 2.67804       & 3763269 & 2272.11   & 3815.03  \\
\bottomrule
\end{tabular}

\end{centering}
\end{table}

In Tables~\ref{table:fairfaxobjective} and \ref{table:alleghenyobjective}, we show how each algorithm performs in terms of each problem's objective function for each dataset respectively. The values are in units of meters for the $k$-medians and $k$-center objective functions, and square meters for the $k$-means objective function.

It is immediately apparent that \textsc{FairKCenter} tends to place more centers in denser regions, compared to other algorithms. This is consistent with the intuition behind the algorithm, as address points in dense regions have relatively smaller neighborhood radius. Although the maximum travel distance is increased significantly relative to the other algorithms, these large travel distances are experienced only by few residents of particularly sparse areas. The increase in average travel distance is more modest, and in Fairfax County our algorithm does even better than the $k$-medians algorithm with respect to the $k$-medians objective function.
In exchange for these compromises, our algorithm does significantly better with respect to $\alpha$, ensuring that no individual will needs to venture too far from their density-dependent neighborhood.

Furthermore, \textsc{FairKCenter} balances the load more evenly across centers than other algorithms. Table~\ref{table:clusterdev} shows the standard deviation in the number of address points served by each center, i.e., the number of points for which that center is the nearest. For both counties, this value is significantly lower for \textsc{FairKCenter} than for the other algorithms. Our algorithm balances load particularly well compared to the $k$-center algorithm, which essentially ignores population density. 

\begin{table}[htp]
\begin{centering}
\caption{Standard deviation in cluster sizes}\label{table:clusterdev}
\begin{tabular}{lll}
\toprule
&\multicolumn{2}{c}{County}\\ \cmidrule{2-3}
Algorithm              & Fairfax & Allegheny \\
\midrule
\textsc{FairKCenter} & 1032.49 & 1696.95\\
$k$-means   & 1344.17 & 2273.53\\
$k$-medians & 1630.06 & 1922.53\\
$k$-center  & 2758.44 & 5691.10 \\
\bottomrule
\end{tabular}

\end{centering}
\end{table}

\section{Conclusion}
We have formulated a simple geometric concept that captures an intuitive notion of fairness: To whatever extent possible, an individual should have access to resources within her own neighborhood. We have proved basic properties of this fairness concept, given a general approximation algorithm for its optimization, and shown that this algorithm performs well on real data.

One potential future direction for this work is to refine the notion of what constitutes an individual's ``neighborhood'' for a given purpose. We used the inverse of local population density as a proxy for the size of a neighborhood, and there are good reasons to believe that these two are correlated. But a more sophisticated approach to defining neighborhood size --- and possibly shape --- might incorporate data on transit times and availability of different modes of transportation. More ambitiously, cellular location data might be used to establish the extent of the common ``orbits'' of residents of a given small area.

On the theoretical side, an obvious next direction is to find algorithms that yield stronger approximation ratios. Our algorithm \textsc{FairKCenter} improved on \textsc{2FairKCenter} in our experiments by less aggressively eliminating candidate centers, but we do not have any theoretical characterization of the instances on which \textsc{FairKCenter} will achieve fairness that is strictly better \textsc{2FairKCenter}. The monotonicity property that is necessary to ensure a fully successful binary search does not hold in general, but one might still be able to identify critical values of $\alpha$ in this range, solve the problem at each of these critical values, and use the smallest one that results in a number of centers that is at most $k$.

Another interesting extension would be to allow Steiner points, removing the restriction that the solution set $S$ is a subset of the population $P$. While \textit{prima facie} it looks like we might have to consider infinitely many such possible centers, one can show that the only points we need to consider as centers are points $x$ such that for some $r$, the boundary of $B_r (x)$ contains at least 3 of the input points. This reduces the number of Steiner points to consider to at most $n^3$.
 
Finally, while we have tight lower and upper bounds on $\alpha$ for arbitrary metric spaces, for Euclidean spaces we have a lower bound of $\sqrt{2}$ and an upper bound of 2. It would be interesting to close the gap, perhaps by improving the upper bound for this case.

\section*{Acknowledgments}

We thank Moni Naor and Omer Reingold for helpful discussions, and we thank Anupam Gupta for alerting us to the similarities between our Theorem~\ref{thm:general} and Lemma 1 of Chan, Dinitz, and Gupta~\cite{ChDiGu06}.

\bibliographystyle{plain}
\bibliography{fairkcenters}

\end{document}